\UseRawInputEncoding
\RequirePackage{fix-cm}

\documentclass[smallextended]{svjour3}       
\smartqed  
\usepackage{graphicx}
\usepackage{xcolor}
\usepackage{wrapfig}
\usepackage{subfig}
\usepackage{rotating}
\usepackage{url}
\usepackage{makeidx}
\usepackage{multirow}
\usepackage{stmaryrd} 
\usepackage{setspace}
\usepackage{algpseudocode}
\usepackage{algorithmicx}
\usepackage{amssymb}
\usepackage[justification=centering]{caption}
\usepackage{multirow}
\usepackage[flushleft]{threeparttable}
\usepackage{scrextend}
\usepackage[english]{babel}
\usepackage{blindtext}
\usepackage{adjustbox}
\usepackage{amsmath}
\usepackage{hyphenat} 
\usepackage{float}
\usepackage{braket}
\usepackage{tabularx}
\begin{document}

\title{Efficient Simulation of Quantum Secure Multiparty Computation
}


\author{Kartick Sutradhar      
}


\institute{Kartick Sutradhar \at
              Indian Institute of Information Technology Sri City \\
              \email{kartick.sutradhar@gmail.com}           
}

\date{Received: date / Accepted: date}

\maketitle

\begin{abstract}
One of the key characteristics of secure quantum communication is quantum secure multiparty computation. In this paper, we propose a quantum secure multiparty summation (QSMS) protocol that can be applied to many complex quantum operations. It is based on the $(t, n)$ threshold approach. We combine the classical and quantum phenomena to make this protocol realistic and secure. Because the current protocols employ the $(n, n)$ threshold approach, which requires all honest players to execute the quantum multiparty summation protocol, they have certain security and efficiency problems. However, we employ a $(t, n)$ threshold approach, which requires the quantum summation protocol to be computed only by $t$ honest players. Our suggested protocol is more economical, practical, and secure than alternative protocols. 

\end{abstract}
\section{Introduction}
Quantum Secure Multiparty Computation (QSMPC) is an advanced cryptographic protocol that leverages the principles of quantum mechanics to enhance security in collaborative computations involving multiple parties. Traditional secure multiparty computation ensures that a group of participants can jointly compute a function over their inputs without revealing the individual inputs to each other. QSMPC goes a step further by incorporating quantum resources, such as entanglement and quantum key distribution (QKD)\cite{sutradhar2021enhanced}, to achieve enhanced security and efficiency. An efficient simulation of QSMPC typically involves combining classical cryptographic techniques with quantum algorithms to mitigate potential vulnerabilities in both realms. The simulation framework must address critical challenges, such as minimizing quantum resource overhead, ensuring robustness against noise in quantum channels, and preserving computational efficiency \cite{song2017t,sutradhar2020efficient,sutradhar2020generalized}. By utilizing quantum gates, entanglement distribution, and secure communication protocols, researchers can simulate QSMPC to validate its practicality for real-world applications \cite{mashhadi2012analysis,mashhadi2012novel,sun2020toward}. Key advantages include stronger security guarantees against quantum adversaries and reduced reliance on computational hardness assumptions \cite{shi2017quantum,sutradhar2021efficient}. Additionally, hybrid classical-quantum methods optimize resource utilization \cite{sutradhar2020hybrid}, making the approach feasible with near-term quantum technologies. Efficient simulations pave the way for deploying QSMPC in sensitive areas \cite{shi2016quantum,shi2016comment,shi2018efficient} like secure voting, financial transactions, and distributed data analysis, where both data privacy and integrity are paramount. These developments highlight the transformative potential of quantum technologies in redefining secure computation paradigms\cite{Qin2018Multidimensional,bao2009threshold,sutradhar2021cost,yang2013secret,mashhadi2017new,hillery1999quantum}.

\section{Preliminaries}
Here, we introduce the Shamir's secret sharing, $QFT$, and $IQFT$, which will be used in our proposed protocol.
\subsection{Shamir's Secret Sharing \cite{shamir1979share}}
This protocol has two phases as discussed below.
\subsubsection{Sharing of Secret}
The dealer creates $n$ shares of the secret using a polynomial $f(x)$ of degree ($t-1$)  and distributes $n$ shares among $n$ participants.
\subsubsection{Reconstruction of Secret}
The threshold number of participants reconstructs the secret as follows.
\begin{equation}\label{equ1}
f(x) = \sum_{v=1}^{t} f(x_v) \prod_{1 \le j \le t, j \neq v} \frac{x_j}{x_j - x_v}
\end{equation}
\subsection{Quantum Fourier Transform ($QFT$) \cite{song2017t,sutradhar2023quantum,sutradhar2021secret,sutradhar20213efficient,sutradhar20244quantum,sutradhar2024smart,challagundla2024privacy,sutradhar2024review,challagundla2024secure,nayaka2024survey}}
The quantum Fourier transform (QFT) is defined as $$QFT: \ket{\alpha}  \rightarrow \frac{1} {\sqrt{d}} \sum_{\beta=0}^{d-1} e^{2\pi i\frac{\alpha}{d}\beta} \ket{\beta}.$$
\subsection{Inverse Quantum Fourier Transform ($IQFT$) \cite{song2017t,sutradhar2024survey,sutradhar2023secure,sutradhar2024svqcp,sutradhar2022privacy,sutradhar20211efficient,sutradhar2024privacy,sutradhar20244privacy,sutradhar2023quantum}}
The inverse quantum Fourier transform ($IQFT$) is defined as $$IQFT: \ket{\beta}  \rightarrow \frac{1} {\sqrt{d}} \sum_{\alpha=0}^{d-1} e^{-2\pi i\frac{\beta}{d}\alpha} \ket{\alpha}.$$
\section{Our Contribution}
In this section, we propose a (t,n) threshold QSMS protocol. Let the dealers $A$ and $B$ have two secrets (for simplicity, we only take two secrets but the secrets can be any number $n$ or more than $n$, where $n$ denotes total no of players) $X$ and $Y$, respectively, and $n$ players want to jointly perform the summation $(S=X+Y)$ without revealing their secrets.  In this protocol, each qualified subset $\mathbb{P}=\{P_1, P_2, \dots, P_t\}$  contains a $k^{th}$ player as an initiator. We assume that $k^{th}$ player is $P_1$, which acts as an initiator. The initiator $P_1$ only contains his share value, nothing else. The process of quantum secure multiparty summation is given as follows.\\
\\
\textbf{Step 1:} $A$ and $B$ choose two distinct $(t-1)$-degree polynomials $f(x)=X + \alpha_1x + \alpha_2x^2 + \dots + \alpha_{t-1}x^{t-1}$ and $g(x)=Y + \beta_1x + \beta_2x^2 + \dots + \beta_{t-1}x^{t-1}$,  $X$ and $Y$ are secrets and the symbol $'+'$ is defined as addition modulo $d$, $d$ is a prime such that $n \le d \le 2n$. The $A$ and $B$ use the Shamir's secret sharing to compute the shares $f(x_i)$ and $g(x_i)$, respectively, which are distributed among $n$ players using an authenticated classical channel. The player $P_i$ only knows the shares $f(x_i)$ and $g(x_i)$, $i= 1, 2, \dots, n$.\\
\\
\textbf{Step 2:} Player $P_i$ computes $h(x_i)= f(x_i) + g(x_i)$, $i= 1, 2, \dots, n,$ and possesses the share $h(x_i)$ only.\\
\\
\textbf{Step 3:} Player $P_u$ computes the shadow $(m_u)$ of the share $h(x_u)$,  $u= 1, 2, \dots, t$, as follows.
\begin{equation}\label{equ4}
m_u  = h(x_u) \prod_{1\leq z\leq t, z\neq u} \frac {x_z} {x_z - x_u} \mod d
\end{equation}
\textbf{Step 4:} Initiator player $P_1$ prepares $t-$particle entangled states as follows.
\begin{equation}
\ket{\Psi_1} = \frac{1} {\sqrt{d}} \sum_{c=0}^{d-1} \ket{c}_1 \ket{c}_2 \dots \ket{c}_t
\end{equation}
Player $P_1$ sends the particle $\ket{c}_u$ to player $P_u$,  $u=2, 3, \dots, t$.\\
\\
\textbf{Step 5:} Each player $P_u$ performs the $QFT$ \cite{Yang2018} on his particle $\ket{c}_u$ as follows:
\begin{equation}\label{equ5}
QFT \ket{c}_1 = \frac{1} {\sqrt{d}} \sum_{a_1=0}^{d-1} e^{2\pi i\frac{c}{d}a_1} \ket{a_1}.
\end{equation}
Each player $P_u$, $(u=1, 2, \dots, t)$, also applies the Pauli operator $U_{m_u, 0}$ on his particle as follows:
\begin{equation}\label{equ6}
U_{m_1, 0} = \sum_{c=0}^{d-1}  \omega^{c.0} \ket{c+m_1} \bra{c}
\end{equation}
After performing the $QFT$ and Pauli operator, the resultant state $\ket{\Psi_2}$ is obtained as follows.
\begin{equation}\label{equ7}
\begin{split}
\ket{\Psi_2} & = U_{m_1, 0} QFT \otimes U_{m_2, 0} QFT \otimes \dots \otimes U_{m_t, 0} QFT \ket{\Psi_1}\\
& = d^{-\frac{t+1}{2}}  \sum_{0 \le a_1, \dots, a_t <d,~ a_1 + ,\dots, + a_t = 0 \mod~d} \ket{a_1+m_1} \ket{a_2+m_2} \dots \ket{a_u+m_u}
\end{split}
\end{equation}
\textbf{Step 6:} Each player $P_u$ performs the measurement operation on his particle $\ket{a_u+m_u}$ in computational basis $\{ \ket{1}, \ket{2}, \dots, \ket{d-1} \}$, and broadcasts his measurement results $a_u+m_u$, where $u=1, 2, \dots, t$.\\
\\
\textbf{Step 7:} Finally, the players in qualified subset 
calculate the summation jointly by summing their results of measurement: $S=\sum_{u=1}^{t} a_u+m_u \mod~d$.
\section{Correctness}
\begin{lemma}
If QFT and Pauli operators are honestly performed by all players in a qualified subset $\mathbb{P}=\{P_1, P_2, \dots, P_t\}$, then they can jointly compute the multiparty quantum summation $(\sum_{u=1}^{t} m_u \mod~d)$ correctly.
\end{lemma}
\begin{proof}
If QFT and Pauli operators are honestly performed by every player in the qualified subset $\mathbb{P}=\{P_1, P_2, \dots, P_t\}$, the quantum state is obtained as follows.
\begin{equation}
\begin{split}
\ket{\Psi_2} & = U_{m_1, 0} QFT \otimes \dots \otimes U_{m_t, 0} QFT \Big(\frac{1} {\sqrt{d}} \sum_{c=0}^{d-1} \ket{c}_1 \dots \ket{c}_t \Big)\\
&=\frac{1} {\sqrt{d}} \sum_{c=0}^{d-1} U_{m_1, 0} QFT \ket{c}_1 \otimes \dots \otimes U_{m_t, 0} QFT \ket{c}_t\\
&=\frac{1} {\sqrt{d}} \sum_{c=0}^{d-1} \Big( U_{m_1, 0} \frac{1} {\sqrt{d}} \sum_{a_1=0}^{d-1} \omega^{a_1c} \ket{a_1}\Big) \otimes\dots \otimes \Big( U_{m_t, 0} \frac{1} {\sqrt{d}} \sum_{a_1=0}^{d-1} \omega^{a_tc} \ket{a_t}\Big)\\
& = d^{-\frac{t+1}{2}} \sum_{0 \le a_1, \dots, a_t <d} \sum_{c=0}^{d-1} \omega^{(a_1 + \dots + a_t)c} \ket{a_1+m_1} \otimes \dots \otimes\ket{a_t+m_t}\\
& = d^{-\frac{t+1}{2}}s_0d\sum_{0 \le a_1, \dots, a_t <d, a_1+ \dots + a_t=0\mod d} \ket{a_1+m_1}  \otimes \dots \otimes \ket{a_t+m_t}
\end{split}
\end{equation}
Each player $P_u$, $u=1, 2, \dots, t$, performs the measurement operation on his own particle in computational basis $\ket{a_u+m_u}$. The QSMS can be computed after receiving the measurement results of each player $P_u$, $u=1, 2, \dots, t$. The QSMS of secret can be calculated as follows.
\begin{equation}\label{equ0000}
\sum_{u=1}^{t} a_u+m_u \overset{d}{\equiv} \sum_{u=1}^{t} a_u + \sum_{u=1}^{t} m_u \overset{d}{\equiv} \sum_{u=1}^{t} m_u \mod~d
\end{equation}
Thus, the multiparty quantum summation of secrets equals to $\sum_{u=1}^{t} m_u \mod~d$.
\end{proof}
\section{Illustration of Secure Multiparty Quantum Summation}
Here, we use a numerical example to discuss the working of the proposed protocol. Let $A$ and $B$ hold two secrets $2$ and $3$, respectively and they want to perform the summation $S=(2+3)$. $A$ and $B$ choose threshold $(t)=3$, total number of players $(n)=7$, and prime $(d)=11$. Suppose $A$ and $B$ select two different polynomials $f(x)=2 + x + x^2 \mod~11$ and $g(x)=3 + x + x^2 \mod~11$, respectively. They calculate the shares $f(x_i)$ and $g(x_i), i=1,2,...,7$ using the Shamir's secret sharing, and allocate these shares to $7$ players. Each player $P_i, i = 1, 2, \dots, 7$, performs $h(x_i)= f(x_i) + g(x_i) \mod~11$. The calculation of shares $h(x_i)$ is shown in Table~\ref{table126}.
\begin{table}[ht]
\centering
\caption{Share Computation}
\begin{tabular}{|l|l|c|c|c|c|c|c|c|}
\hline
\multicolumn{2}{|l|}{Players} & $P_1$ & $P_2$ & $P_3$  & $P_4$   & $P_5$    & $P_6$      & $P_7$       \\ \hline
\multirow{4}{*}{\rotatebox[origin=c]{90}{~~~~Shares}}  
                        & $f(x_i)$   & 4  & 8  & 3   & 0    & 10     & 0       & 3        \\ \cline{2-9} 
                        & $g(x_i)$   & 5  & 9  & 4   & 1    & 0     & 1       & 4        \\ \cline{2-9} 
                        & $h(x_i)$   & 9  & 6 & 7   & 1    & 10    & 1       & 7        \\ \cline{1-9} 
                       
\end{tabular}
\label{table126}
\end{table}
Each player $P_u$, u= 1, 2, 3, computes the shadow of the shares $m_u$, as $m_1  = 9. \Big(\frac {2} {2 - 1} . \frac {3} {3 - 1} \Big) \mod 11 =5$, $m_2  = 6. \Big(\frac {1} {1 - 2} . \frac {3} {3 - 2} \Big) \mod 11 = 4$, and $m_3  = 7. \Big(\frac {1} {1 - 3} . \frac {2} {2 - 3} \Big) \mod 11 =7$, respectively (using Eq.~\ref{equ4}).
The player $P_1$ now computes $\ket{\Psi_1} = \frac{1} {\sqrt{11}} \sum_{c=0}^{10} \ket{c}_1 \ket{c}_2 \ket{c}_3$ and sends the particle $\ket{c}_u$ to player $P_u, u=2, 3$. Each player $P_u, u=1, 2, 3,$ applies the $QFT$ and  Pauli operator $U_{5, 0}$, $U_{4, 0}$, $U_{7, 0}$ on his particle, respectively, (as per Eq.~\ref{equ7}).
\begin{equation}
\begin{split}
\ket{\Psi_2} & = U_{5, 0} QFT \otimes U_{4, 0} QFT \otimes U_{7, 0} QFT \Big(\frac{1} {\sqrt{11}} \sum_{c=0}^{10} \ket{c}_1 \ket{c}_2 \ket{c}_3 \Big)\\
&=\frac{1} {\sqrt{11}} \sum_{c=0}^{10} U_{5, 0} QFT \ket{c}_1 \otimes U_{4, 0} QFT \ket{c}_2 \otimes  U_{7, 0} QFT \ket{c}_3\\
& = 11r_1 \sum_{0 \le a_1, a_2, a_3 <10,~ a_1 + a_2 + a_3 = 0 \mod~11} \ket{a_1+5} \ket{a_2+4} \ket{a_3+7}
\end{split}
\end{equation}
Each player $P_u, u = 1,2,3,$ performs the measurement operation in computational basis on his particle. The players $P_1$, $P_2$, and $P_3$ broadcast the measurement results $a_1+5$, $a_2+4$, and $a_3+7$, respectively. Finally, they get the summation by summing the results of measurement as follows:
$$ a_1+5+a_2+4+a_3+7 \overset{11}{\equiv} a_1+a_2+a_3+16\overset{11}{\equiv} 16 \mod 11=5$$
\section{Simulation Results}
We simulate the proposed protocol using the IBM real quantum processor, which is available at T.J.Watson lab, USA. The Hadamard gate is taken as the $QFT$ in this circuit diagram of QSMS. On his particle, the player $P_u$ applies the $QFT$ and also performs the Pauli operator on his particle. Then, each player $P_u$ performs measurement operations on his own particle, and broadcasts the measurement result. Finally, by summing their measurement results, the players jointly calculate the QSMS. The privacy of this protocol is guaranteed until a certain number of players disclose their shares. We have simulated this circuit of QSMS with $3$ players, $5$ qubits, and $8192$ number of average shots. Initially, the player $P_u$, $u=1,2,3$ performs the $QFT$ on his particle $\ket{c}_u$ and also executes the Pauli operator on particle $\ket{c}_u$. Then, each player $P_u, u = 1,2,3,$ executes the measurement operation in computational basis on his particle. The players $P_1$, $P_2$, and $P_3$ broadcast the measurement results $a_1+5$, $a_2+4$, and $a_3+7$, respectively. Finally, they get the summation of $2$ and $3$ by adding the measurement results as follows: $$a_1+5+a_2+4+a_3+7=16 \mod 11=5.$$ The simulation result of the proposed summation protocol for $3$ players, $5$ qubits, and $8192$ number of average shots. The state $101$ (i.e., binary representation of $5$) is calculated efficiently. 
\section{Discussion}
Here, we address the security and performance analysis based on some properties of the proposed QSMS protocol.
\subsection*{Security Analysis}
In this section, we analyze the security of QSMS protocol based on the intercept-resend, entangle-measure, intercept, collective, coherent, and collusion attacks.\\
\\
\textbf{Intercept-resend attack:} Suppose an attacker Mallory intercepts the particle $\ket{c}_u$. It measures the quantum particle $\ket{c}_u$ in the computational basis to get the useful data about the share's shadow ($m_u$). Mallory produces the clone quantum particle $\ket{\bar{c}}_u$ and resends this clone particle to  player $P_u$, $u=2,3,\dots t$. If Mallory applies this method to attack, then it can get $c$ accurately with  probability $\frac{1}{d}$. But, from this attack, Mallory cannot get any useful data about the share's shadow $m_u$, because the intercepted particle $\ket{c}_u$ does not contain any useful data about the share's shadow $m_u$.\\
\\
\textbf{Entangle-Measure attack:} After the intercept attack,  Mallory performs the complex entangle-measure attack on the entangled quantum particle $\ket{c}_u$. In this attack, Mallory performs the measurement operation on the intercepted entangled quantum particle $\ket{c}_u$ in the computational basis to get the useful data about the share's shadow $m_u$. If Mallory applies the entangle-measure attack, then it can get $c$ accurately with probability $\frac{1}{d}$. But, from this attack, Mallory cannot get useful data about the share's shadow $m_u$, because the intercepted entangled quantum particle $\ket{c}_u$ does not contain any useful data about the share's shadow $m_u$.\\
\\
\textbf{Intercept attack:} Suppose Mallory intercepts the particle $\ket{c}_u$ and measures the quantum particle $\ket{c}_u$ in the computational basis to reveal the useful data about the share's shadow $m_u$. If Mallory measures the quantum particle $\ket{c}_u$ in the computational basis, then it can get $c$ correctly with probability $\frac{1}{d}$. But, from the measurement result $c$, it cannot get any useful data about the share's shadow $m_u$, because the intercepted particle $\ket{c}_u$ does not carry any useful data about the share's shadow $m_u$.\\
\\
\textbf{Collective attack:}
In a collective attack, Mallory prepares an autonomous ancillary particle to communicate with each qudit to get the shadow of share and they perform the joint measurement operation on every ancillary qudit. Suppose Mallory communicates with every qudit of all players by preparing an autonomous ancillary particle $\ket{e}$. After successful interaction, Mallory gets the particle $\ket{o}_x$. Then, Mallory wants to know the shadow of share by performing a  computational basis $\{ \ket{1}, \ket{2}, \dots, \ket{d-1} \}$ joint measurement operation. Mallory cannot get any useful data about the share's shadow from this joint measurement operation because $\ket{o}_x$ does not contain any useful data about the share's shadow.\\
\\
\textbf{Coherent attack:}
In this attack, Mallory prepares an autonomous ancillary particle $\ket{c}$ to communicate with the qudits of each player. After interacting, Mallory gets each player's particle $\ket{o}_x$ and performs a joint measurement operation on all players particle $c$ in computational basis $\{ \ket{1}, \ket{2}, \dots, \ket{d-1} \}$. Mallory only gets $o$ from the joint measurement result of particle $\ket{o}_x$ with probability $\frac{1}{d}$. But, the joint measurement result $o$ does not contain any useful data about the share's shadow. From this attack, Mallory only gets the interacting particle $\ket{o}_x$, but it cannot learn any useful data about the share's shadow.\\
\\
\textbf{Collusion attack:} In this protocol, each player $P_u$ performs the measurement on his own particle $\ket{a_u+m_u}$ and broadcasts his result of the measurement $a_u+m_u$,  $u=1, 2, \dots, t$. From this broadcast, other players cannot get any useful data about the share's shadow $m_u$. If some rational players $P_{l-1}$ and $P_{l+1}$ jointly want to get the  data about the share's shadow but they cannot get any useful data about the share's shadow $m_u$ because the initiator $P_1$ transmits only particles $\ket{c}_u$ to all other players and unfortunately $\ket{c}_u$ does not contain any useful data about the share's shadow $m_u$.
\section{Conclusion}
We have examined a $(t, n)$ threshold QSMS protocol based on secret sharing in this paper. If a certain $t$ number of players are honest, this protocol can be carried out effectively. Because it uses linear communication and secret-by-secret computation, it is both efficient and safe. Because the share of secrets is calculated using linear secret sharing, it can also calculate the QSMS if there are more secrets than players. Because we have effectively simulated this protocol using an IBM quantum computer that yields efficient results after increasing the number of shots, this QSMS protocol is more realistic than the current multiparty quantum summation protocols. 

\end{document}